\documentclass[a4paper,12pt]{article}	
\usepackage{amsfonts}							
\usepackage{amssymb}							
\usepackage{amsmath}							
\usepackage{bbm}								
\usepackage{bm}									
\usepackage{mathrsfs}							
\usepackage{eucal}								
\usepackage{color}								
\usepackage{fancyhdr}							
\usepackage[dvips,pdftex]{graphicx}		
\usepackage[ pdftex, plainpages = false, pdfpagelabels,
                 pdfpagelayout = SinglePage,
                 bookmarks,
                 bookmarksopen = true,
                 bookmarksnumbered = true,
                 breaklinks = true,
                 linktocpage,
                 pagebackref,
                 colorlinks = true,
                 linkcolor = blue,	        
                 urlcolor  = red,			
                 citecolor = green,			
                 anchorcolor = green,     	
                 hyperindex = true,
                 hyperfigures
                 ]{hyperref} 					

\usepackage{url}							
\newcommand{\al}			{\alpha}

\newcommand{\eps}			{\epsilon}

\newcommand{\tta}			{\theta}

\newcommand{\lm}			{\lambda}

\newcommand{\sgm}		{\sigma}

\newcommand{\Gm}			{\Gamma}

\newcommand{\Tta}			{\Theta}

\newcommand{\Ups}		{\Upsilon}
\newcommand{\Om}			{\Omega}

\newcommand{\mc}[1]{\mathcal{ #1}}						
\newcommand{\trm}[1]{\textrm{ #1}}						

\newcommand{\wt}[1]{\widetilde{ #1}}						
\newcommand{\wht}[1]{\widehat{ #1}}						
\newcommand{\ovl}[1]{\overline{ #1}}						

\newcommand{\ud}{\mathrm{d}} 								

\newcommand{\R}{\mathbbm{R}}								

\newcommand{\dnd}[3]{\frac{\partial^{#1} #2}{\partial #3^{#1}}}		

\newtheorem{theorem}{Theorem}[section]					
																				
\newtheorem{lemma}[theorem]{Lemma}						
																				
\newenvironment{proof}[1][Proof.]{\begin{trivlist}					
\item[\hskip \labelsep {\emph { #1}}]}{\end{trivlist}}

\newcommand{\qed}{\nobreak \ifvmode \relax \else			
      \ifdim\lastskip<1.5em \hskip-\lastskip							
      \hskip1.5em plus0em minus0.5em \fi \nobreak
      \vrule height0.75em width0.5em depth0.25em\fi}
      
\newcommand{\trian}{\nobreak \ifvmode \relax \else			
      \ifdim\lastskip<1.5em \hskip-\lastskip
      \hskip1.5em plus0em minus0.5em \fi \nobreak
      \vrule height0.75em width0.5em depth0.25em\fi}

\newcommand{\citte}[1]{~\cite{#1}}															

\makeatletter																								
\def\url@leostyle{%
  \@ifundefined{selectfont}{\def\UrlFont{\sf}}{\def\UrlFont{\small\ttfamily}}}
\makeatother

\newcommand{\haux}{\mathcal{H}_{\!\!{\trm{aux}}}}							
\newcommand{\thaux}{{\widetilde{\mathcal{H}}_{\!\!{\trm{aux}}}}}		
\newcommand{\Ihaux}{\mathcal{H}_{c}}							 
\renewcommand{\c}{\mc{C}}												
\newcommand{\p}{\mc{P}}												
\newcommand{\bp}{\ovl{\mc{P}}}										
\newcommand{\bc}{\ovl{\mc{C}}}										
\newcommand{\aux}[2]{\left( #1, #2\right)_{\!\!\trm{aux}}}				
\newcommand{\taux}[2]{\left( #1, #2\right)_{\wt{\text{aux}}}}			
\newcommand{\Iaux}[2]{\left( #1, #2\right)_c}									
\newcommand{\Raux}[2]{\left( #1, #2\right)_{\R}}								
\newcommand{\brs}[2]{\left( #1, #2\right)_{\text{BRST}}}					
\newcommand{\brsren}[2]{\left( #1, #2\right)^{r}_{\text{BRST}}}		
\newcommand{\raqsesq}[2]{\left( #1, #2\right)_{\text{ave}}}					
\newcommand{\traqsesq}[2]{\left( #1, #2\right)_{\wt{\text{ave}}}}			
\newcommand{\ipraq}[2]{\left( #1, #2\right)_{\text{RAQ}}}

\newcommand{\Hraq}{{\mathcal{H}_{\text{RAQ}}}}
\newcommand{\cutoff}{L}

\newcommand{\dndl}[2]{\frac{\partial^{\, l} #1}{\partial #2}}			
\newcommand{\ghalt}[1]{{\text{gh}}(#1)}											
\newcommand{\Aobs}{{\mathcal{A}_\mathrm{obs}}}
\newcommand{\BbbR}{\mathbb{R}}
\newcommand{\BbbZ}{\mathbb{Z}}

\newcommand{\mcc}[1]{\CMcal{ #1}}												

\numberwithin{equation}{section}				

\title{Constraint rescaling in refined algebraic quantisation:\\
momentum constraint}
\author{Jorma Louko\thanks{jorma.louko@nottingham.ac.uk}
\ and
Eric Mart\'inez-Pascual\thanks{pmxem2@nottingham.ac.uk}
\\
\noalign{\vspace{3ex}}
\small{\it School of Mathematical Sciences,
University of Nottingham,}\\
\small{\it Nottingham NG7 2RD, UK}
\\
\noalign{\vspace{3ex}}
\small{Revised November 2011}
\\
\noalign{\vspace{1ex}}
\small{Published in J.\ Math.\ Phys.\ 
\textbf{52}, 123504 (2011)}
\\ 
{\small\tt \href{http://link.aip.org/link/?jmp/52/123504}{http://link.aip.org/link/?jmp/52/123504}}
}

\date{}

\begin{document}

\maketitle

\begin{abstract}
We investigate refined algebraic quantisation within 
a family of classically equivalent constrained Hamiltonian systems 
that are related to each other by rescaling a momentum-type constraint. 
The quantum constraint is implemented by a rigging map that is 
motivated by group averaging but has a resolution finer than 
what can be peeled off from the formally divergent contributions 
to the averaging integral. 
Three cases emerge, depending on the asymptotics of the rescaling function: 
(i) quantisation is equivalent to that with identity scaling; 
(ii) quantisation fails, 
owing to nonexistence of self-adjoint extensions of the 
constraint operator; 
(iii) a quantisation ambiguity arises from the self-adjoint 
extension of the constraint operator, and the resolution of 
this purely quantum mechanical ambiguity determines the 
superselection structure of the physical Hilbert space. 
Prospects of generalising the analysis to systems with 
several constraints are discussed.
\end{abstract}

\vspace{4ex}

\noindent 
{\scriptsize Copyright (2011) American Institute of Physics. 
This article may be downloaded for personal use only. 
Any other use requires prior permission of the 
author and the American Institute of Physics.}

\newpage

\section{Introduction}
\label{sec:intro} 

In a classical Hamiltonian system, a gauge symmetry is generated by
constraint functions known as \emph{first class constraints:\/}
constraints whose Poisson brackets with each other and with the
Hamiltonian are linear combinations of the constraints themselves. In
the Dirac-Bergmann quantisation scheme the constraint functions are
promoted into quantum constraint operators, and the physical quantum
states are required to be annihilated by the quantum
constraints\citte{dirbk-lec,sunbk,henbk}.

To find physical quantum states, one may wish to start from a state
that is not necessarily annihilated by the quantum constraints and
average this state over the action generated by the quantum
constraints\citte{hig91a,hig91b}. When the quantum constraints
generate the action of a genuine Lie group, this group averaging
proposal can be given a precise implementation in the framework known
in physics as refined algebraic quantisation
(RAQ)\citte{ash95,emb98,mar00,Giulini:1999kc} and in mathematics as
Rieffel induction\citte{Rieffel1974176}, with results on both
uniqueness and generality of the resulting quantum
theory\citte{giu99b,giu99a}. Case studies of specific quantum
mechanical systems can be found
in\citte{gom98,Louko:1999tj,Louko:2003cn,Louko:2004zq,lou05},
applications to de~Sitter invariant quantum field theory are
considered in\citte{hig91a,hig91b,Marolf:2008it,Marolf:2008hg} and
applications to loop quantum gravity are considered
in\citte{Dittrich:2004bq,Kaminski:2009qb}.

A~Lie group action generated by the quantum constraints is however a
very special case: it can be expected to occur only when the Poisson
brackets of the classical constraint functions form a Lie algebra,
that is, close with structure coefficients that are constants on the
phase space. In many systems of interest, including general relativity
in both metric and connection formulations\citte{Ashtekar:1991hf}, the
structure coefficients are nonconstant functions on the phase
space. Further, given a system with at least two constraints and
constant structure coefficients, redefining the constraints by an
invertible linear map that is not constant on the phase space yields a
classically equivalent system that can be arranged to have nonconstant
structure functions. The distinction between structure constants
(known as a \emph{closed\/} gauge algebra) and nonconstant structure
functions (known as an \emph{open\/} gauge algebra) is hence not
intrinsic to the true physical degrees of freedom but depends also on
how the generators of the gauge transformations are
chosen\citte{kuc86a,kuc86b,haj90,mcm89a,mcm89b}. These considerations
show that there would be considerable interest to extend the group
averaging method to systems with open algebras.

A~proposal for extending group averaging to open gauge algebras 
has been given by Shvedov\citte{shv02}, using the 
Becchi-Rouet-Stora-Tyutin (BRST) formalism\citte{bec76,kug79,henrep,govbk} 
and building on the previous 
work in\citte{raz90,Nirov:1994mj,%
Marnelius:1990eq,Marnelius:1993ba,bat95,mar95,dut99}, 
in particular on the Batalin-Marnelius inner product\citte{bat95}. 
When the structure functions are constants, 
Shvedov's proposal duly reduces to averaging over a 
Lie group in the measure adopted in\citte{giu99a}. 
To recover a full quantum theory, however, an averaging formula must
be supplanted with additional structure, including the state space on
which the averaging acts and the sense in which the averaging
converges. These issues have proven quite delicate already in the Lie
group context when the group is not compact, despite the control
provided by the Giulini-Marolf uniqueness theorem\citte{giu99a}; for
example, the averaged states can turn out to have negative norm
squared\citte{Louko:2003cn}.

In this paper we address group averaging in refined algebraic
quantisation in a class of systems related by rescaling a classical
constraint\citte{christo90}.  We focus on a system with a single
constraint, so that the gauge algebra is trivially closed 
regardless the scaling of the constraint.   
To avoid built-in topological complications in the classical
theory, we take the phase space to be $T^*\BbbR^{2} \simeq \BbbR^4$
and the constraint to be linear in one of the momenta, but we allow
this momentum to be scaled by a nowhere-vanishing function of the
coordinates. The classical reduced phase space is then just $T^*\BbbR
\simeq \BbbR^2$, obtained by dropping the canonical pair whose
momentum appears in the constraint.  The main issue that remains in
quantisation is then how to promote the classical constraint into an
operator in terms of which the quantum gauge transformations and the
averaging over these transformations can be defined.

We shall see that once the auxiliary Hilbert space structure is
specified, the options to define the constraint operator depend on the
asymptotics of the scaling function in the classical constraint. Three
cases emerge:
\begin{itemize}
\item[(i)] The constraint operator is essentially self-adjoint, and
the quantisation is equivalent to the group averaging that arises when
the scaling function is the constant function~$1$.
\item[(ii)]
The constraint operator has no self-adjoint extensions, 
and we are unable to extract a notion of quantum gauge transformations, 
let alone a definition of averaging over them. 
No quantum theory is recovered. 
\item[(iii)] There is an infinite quantisation ambiguity, arising from
a choice in the self-adjoint extension of the constraint
operator. Within a subclass of extensions parametrised by one smooth
function of one variable, the superselection structure of the physical
Hilbert space depends strongly on the choice of the extension, but the
quantum theory is insensitive to the residual freedom in the scaling
function.
\end{itemize}

The superselection sectors that emerge in case (iii) resemble 
closely those in refined algebraic quantisation of the 
Ashtekar-Horowitz-Boulware model\citte{lou05}. However, 
whereas with the Ashtekar-Horowitz-Boulware model these 
sectors are determined by the potential term in the classical constraint, 
here the sectors are determined solely by a quantisation ambiguity. 

We begin by introducing the classical system in
section~\ref{the-model}. Section \ref{generic-M} specifies the
auxiliary structure for refined algebraic quantisation, establishing
the conditions under which the scaling functions belong to cases
(i)--(iii). Case (i) is briefly addressed in section~\ref{case-i}.
The main content of the paper is in the analysis of case (iii) in
section~\ref{case-iii}. Section \ref{remarks} presents a summary and
concluding remarks. Appendix \ref{brst-analysis} reviews the
relationship of group averaging and the BRST inner product for a
system with a single constraint. The proofs of certain technical
results are deferred to appendices \ref{app:rigging}
and~\ref{lemmas-app}.

We set $\hbar=1$.  Complex conjugate is denoted by overline, except in
appendix \ref{brst-analysis} where it is denoted by~${}^*$.  In
asymptotic analyses, $O(u)$ is such that $u^{-1}O(u)$ remains bounded
as $u\to0$, $o(u)$ is such that $u^{-1}o(u) \to 0$ as $u\to0$ and
$o(1) \to 0$ as $u\to0$\citte{wongbook}.

\section{Classical system: one momentum-type constraint}
\label{the-model}

We consider a system with configuration space 
$\BbbR^2 = \{(\tta,x)\}$ and phase space 
$\Gm = T^*\BbbR^2 = \{(\tta,x,p_{\tta},p_x)\} \simeq \BbbR^4$. 
The system has one constraint, 
\begin{align}
\label{rc}
\phi := M(\tta,x)p_{\tta}, 
\end{align}
where the real-valued function $M$ is smooth and nowhere vanishing.
We may assume without loss of generality that $M$ is positive.  We
assume that there is no true Hamiltonian, although inclusion of a true
Hamiltonian that only depends on $x$ and $p_x$ would be
straightforward. 

The constraint hypersurface $\phi=0$ is 
$\Gm_{c}=\{(\tta, x,0,p_{x})\} \simeq \BbbR^3$. 
The generator of gauge transformations on $\Gm_{c}$
is the restriction of the 
Hamiltonian vector field of~$\phi$, 
\begin{align}
\label{eq:Xplus-def}
X^+ := M(\tta,x)\partial_\theta . 
\end{align}
As $X^+$ is nowhere vanishing, 
the constraint is regular in 
the sense of\citte{henbk,mis03}. 
The integral curves of $X^+$ 
have constant $x$ and $p_x$
and they connect any two given values of~$\theta$. 
The reduced phase space is hence 
$\Gm_{\text{red}}=\{(x,p_{x})\} \simeq \BbbR^2$. 

If we wish to view the gauge transformations
as maps on~$\Gm_{c}$, 
rather than just as maps of individual 
initial points in~$\Gm_{c}$, a subtlety arises. 
The gauge transformation with the 
(finite) parameter $\lm$ is the 
exponential map of~$\lm X^+$, $\exp(\lm X^+)$. 
If $M$ satisfies 
\begin{align}
\label{eq:completeness-cond}
\int_{-\infty}^0 \frac{\ud \tta}{M(\tta,x)} 
= \infty = 
\int_0^\infty \frac{\ud \tta}{M(\tta,x)}  
\end{align}
for all~$x$, then $X^+$ is a complete vector field, 
and the family $\{\exp(\lm X^+) \mid \lm\in\BbbR\}$ is a one-parameter 
group of diffeomorphisms $\Gm_{c} \to \Gm_{c}$\citte{leebk}. 
If \eqref{eq:completeness-cond} does not hold for all~$x$, 
then $X^+$ is incomplete. 
It is still true that the action of 
$\exp(\lm X^+)$ 
on any \emph{given\/} 
initial point in $\Gm_{c}$  
is well defined for sufficiently small~$|\lm|$; however, 
there are no values of $\lm\ne0$ for which both of $\exp(\pm\lm X^+)$ 
are defined as maps $\Gm_{c} \to \Gm_{c}$, 
since at least one of them will try to move points past the infinity. 
It is this classical subtlety whose quantum mechanical counterpart 
will be at the heart of our quantisation results. 

Finally, note that when $M$ is the constant function~$1$, 
we have $\phi=p_\tta$ and $X^+ = \partial_\tta$, 
and the gauge transformation 
$\exp(\lm X^+): \Gm_{c} \to \Gm_{c}$ is just the translation 
$(\tta,x,p_x) \mapsto (\tta+\lm,x,p_x)$. 
Other choices for $M$ amount to rescaling the constraint 
of this prototype system by a positive function that may depend on 
both the gauge variable $\theta$ and the non-gauge variable~$x$. 
We refer to $M$ as the scaling function.

\section{Refined algebraic quantisation: action of the gauge group}
\label{generic-M}

We wish to quantise the system in the refined algebraic quantisation
(RAQ) framework as reviewed in\citte{mar00}. In this section we
specify the auxiliary structure and examine conditions under which the
quantum constraint generates the action of a unitary group on the
auxiliary Hilbert space. Textbook expositions of the requisite theory
of self-adjoint operators are given
in\citte{reebk2,blabk} and a pedagogical introduction can be
found in\citte{bonneau-etal}. 

We take the auxiliary Hilbert space to be square integrable functions
on the classical configuration space 
$\BbbR^2 = \{(\tta,x)\}$, 
$\haux:=L^{2}(\BbbR^{2},\ud\tta \, \ud x)$. 
The inner product in $\haux$ reads  
\begin{align}
\label{aux-ip}
\aux{\psi}{\chi}:=\int_{\BbbR^{2}} \ud\tta \, \ud x\; 
\ovl{\psi(\tta,x)} \, \chi(\tta,x) , 
\end{align}
where the overline denotes complex conjugation. 

We promote the classical constraint 
$\phi$ \eqref{rc} into a quantum constraint 
by the substitution $p_\tta \mapsto -i\partial_\tta$ and a 
symmetric ordering, with the result
\begin{align}
\label{eq:phihat-def-bos}
\wht{\phi} 
& :=-i\left(M\partial_{\tta}
+\tfrac{1}{2}(\partial_{\tta}M)\right) . 
\end{align}
We wish to obtain a family of operators
$\{U(\lm)\}$ by exponentiating~$\wht{\phi}$,
\begin{align}
\label{eq:U-formaldef}
{U}(\lm) &:= \exp\bigl(i\lm\wht{\phi}\,\bigr) , 
\end{align}
and to find an inner product on the physical
Hilbert space by a suitable interpretation of the 
sesquilinear form 
\begin{align}
\label{eq:raqsesq}
\raqsesq {\psi}{\chi} 
&:= 
\int \ud\lm 
\aux{\psi}{{U}(\lm)\chi} . 
\end{align}
In this section we consider~\eqref{eq:U-formaldef}. 

The operator $\wht{\phi}$ \eqref{eq:phihat-def-bos} is symmetric on
the dense linear subspace of smooth functions of compact support
in~$\haux$. If $\wht{\phi}$ has self-adjoint extensions on~$\haux$, a
choice of the self-adjoint extension in \eqref{eq:U-formaldef} defines
$\{{U}(\lm) \mid \lm\in\BbbR\}$ as a one-parameter group of unitary
operators, and we can seek to implement 
\eqref{eq:raqsesq} as the group averaging sesquilinear form
in~RAQ\null. We hence need to analyse the self-adjoint extensions
of~$\wht{\phi}$.

The existence of self-adjoint extensions of $\wht{\phi}$ 
is determined by
the deficiency indices of~$\wht{\phi}$, 
that is, the dimensions 
of the linear subspaces of $\haux$ satisfying 
$\wht{\phi}\psi=\pm i \psi$\citte{reebk2,blabk}. 
The solutions to the differential equation 
$\wht{\phi}\psi=\pm i \psi$ are 
\begin{align}
\label{eq:psipm}
\psi_{\pm}(\tta,x)=\dfrac{F_{\pm}(x)}{\sqrt{M(\tta,x)}}
\exp\bigl[\mp \sgm_{x}(\tta)\bigr],
\end{align}
where 
\begin{align}
\label{sgm-of-theta}
\sgm_{x}(\tta):=\int_{\, 0}^{\, \tta}\frac{\ud\tta'}{M(\tta',x)} 
\end{align}
and the complex-valued functions $F_{\pm}$ are arbitrary. 
There are three qualitatively different cases, 
depending on the asymptotics of $\sgm_{x}(\tta)$ as $\theta \to \pm\infty$. 


\emph{Type I scaling functions.}\ \ 
Suppose that 
\begin{align}
\label{eq:sigma-inftyconditions}
\sgm_{x}(\tta) \to \pm\infty 
\ \ \text{as}\ \ 
\theta\to\pm\infty 
\ \ \text{for a.e.~$x$}, 
\end{align} 
where ``a.e.'' stands for almost everywhere 
in the Lebesgue measure on~$\BbbR$. 
Then every nonzero $\psi_{\pm}$ \eqref{eq:psipm} has infinite norm, 
for $\psi_{+}$ because of the behaviour at $\tta\to-\infty$ and for 
$\psi_{-}$ because of the behaviour at $\tta\to\infty$. 
The deficiency indices are $(0,0)$ and $\wht{\phi}$ 
is essentially self-adjoint. The operator ${U}(\lm)$ is unitary, 
and it acts on the wave functions by the 
exponential map of the vector field $X^+$~\eqref{eq:Xplus-def}, 
where the wave functions are understood as half-densities (see for example 
Appendix C in\citte{Ashtekar:1991hf}). 
Explicitly, we have 
\begin{align}
\label{formal-Uaction}
\bigl({{U}}(\lm)\psi\bigr)
(\tta,x)
=
\frac{\sqrt{M\bigl(\sgm_{x}^{-1}(\sgm_{x}(\tta)+\lm),x\bigr)}}
{\sqrt{M(\tta,x)}}\,\psi\bigl(\sgm_{x}^{-1}(\sgm_{x}(\tta)+\lm),x\bigr), 
\end{align}
where the formula is well-defined for all $x$ 
except the set of 
measure zero (if non-empty) where 
\eqref{eq:sigma-inftyconditions} does not hold. 
The group multiplication law in the one-parameter group 
$\{{U}(\lm) \mid \lm\in\BbbR\} \simeq \BbbR$ is addition in~$\lm$. 
In the special case $M(\tta,x)=1$, we recover 
the group of translations in~$\tta$, 
$\bigl({{U}}(\lm)\psi\bigr)
(\tta,x) = \psi(\tta+\lm,x)$. 


\emph{Type II scaling functions.}\ \ 
Suppose that \eqref{eq:sigma-inftyconditions} holds 
either with the upper signs or with the lower signs but not both. 
If \eqref{eq:sigma-inftyconditions} holds for the upper signs, 
then every nonzero $\psi_{-}$ \eqref{eq:psipm} has 
again infinite norm; however, any $F_{+} \in L^2(\BbbR, \ud x)$ 
whose support is in the set 
where \eqref{eq:sigma-inftyconditions} 
with the lower signs fails 
will give a square integrable~$\psi_{+}$. 
The deficiency indices are hence $(\infty,0)$. Similarly, 
if \eqref{eq:sigma-inftyconditions} holds for the lower signs, 
the deficiency indices are $(0,\infty)$. 
$\wht{\phi}$~has no self-adjoint 
extensions in either case, and 
\eqref{eq:U-formaldef} does 
not provide a definition of~${U}(\lm)$. 
At the level of formula~\eqref{formal-Uaction}, 
the problem is that $\sgm_{x}^{-1}$ is 
not defined even for a.e.~$x$. 


\emph{Type III scaling functions.}\ \ 
Suppose that \eqref{eq:sigma-inftyconditions} holds 
with neither upper nor lower signs. 
Reasoning as with Type II above shows that the 
deficiency indices are $(\infty,\infty)$. 
$\wht{\phi}$ has an infinity of self-adjoint extensions, 
and each of them defines 
$\{{U}(\lm) \mid \lm\in\BbbR\}$
as a one-parameter 
group of unitary operators. 
Formula \eqref{formal-Uaction} has again 
a problem in that $\sgm_{x}^{-1}$ is not defined, 
but the self-adjoint extension of $\wht{\phi}$ 
provides a rule by which the probability that 
is pushed beyond $\theta=\pm\infty$ by 
\eqref{formal-Uaction}
will re-emerge from  $\theta=\mp\infty$. 
The group $\{{U}(\lm) \mid \lm\in\BbbR\}$ 
may be isomorphic 
to either $\BbbR$ or~$\text{U}(1)$. 

We are hence able to proceed only with Types I and~III\null. 
In sections \ref{case-i} and \ref{case-iii} we address the integral
\eqref{eq:raqsesq} for these two types.

\section{RAQ for Type I scaling functions}
\label{case-i}

For Type I scaling functions, the
multiplication law in the group 
$\{{U}(\lm) \mid \lm\in\BbbR\} \simeq \BbbR$ 
is addition in~$\lm$. 
We hence take the range of integration in 
\eqref{eq:raqsesq} to be the full real axis. 

It is convenient to map $\haux$ into 
$\thaux := L^{2}(\BbbR^{2},\ud\Tta \, \ud x)$ 
by the Hilbert space isomorphism 
\begin{align}
\haux & \to \thaux , 
\notag
\\
\psi & \mapsto \wt{\psi} , 
\notag
\\
\wt{\psi}(\Tta,x) & := 
\sqrt{M\bigl(\sgm_{x}^{-1}(\Tta),x\bigr)} 
\, \psi\bigl(\sgm_{x}^{-1}(\Tta),x\bigr) , 
\end{align}
where the last line is well defined for a.e.~$x$.  
Working in $\thaux$, the auxiliary inner product
reads 
\begin{align}
\label{aux-ipt}
\taux{\wt\psi}{\wt\chi}:=\int_{\BbbR^{2}}  \ud \Tta \, \ud x \; 
\ovl{{\wt\psi}(\Tta,x)} \, \wt\chi(\Tta,x) , 
\end{align}
and the 
group averaging sesquilinear form takes the form 
\begin{align}
\label{eq:raqsesqt}
\traqsesq{\wt\psi}{\wt\chi} 
&:= 
\int_{-\infty}^\infty \ud\lm
\, \taux{\wt\psi}{{\wt{U}}(\lm)\wt\chi} , 
\end{align}
where 
\begin{align}
\label{tUaction}
\Bigl({{\wt{U}}}(\lm)\wt\psi \, \Bigr)
(\Tta,x)
=
\wt\psi\bigl(\Tta+\lm,x\bigr). 
\end{align}
The system has thus been mapped to that in which $M$ is the constant
function~$1$.

RAQ in $\thaux$ can now be carried out as
for the closely related system discussed in Section IIB
of\citte{ash95}. We can choose smooth functions of compact support on
$\BbbR^2 = \{(\Tta,x)\}$ as the dense linear subspace of $\thaux$ on
which \eqref{eq:raqsesqt} is well defined. The averaging projects out
the $\Tta$-dependence of the wave functions, and the physical Hilbert
space is $L^2(\BbbR, \ud x)$. The technical steps are 
essentially identical to those
in\citte{ash95} and we will not repeat them here.

\section{RAQ for Type III scaling functions}
\label{case-iii} 

For Type III scaling functions, an attempt to classify the
self-adjoint extensions of $\wht{\phi}$ would face two
challenges. First, the sets in which the conditions
\eqref{eq:sigma-inftyconditions} fail for the upper and lower signs
can be arbitrary sets of positive measure. Second, even after these
sets are fixed, the deficiency indidices are $(\infty,\infty)$, and
the self-adjoint extensions of $\wht{\phi}$ comprise only a subset of
all maximal extensions of~$\wht{\phi}$\citte{blabk}.  We shall
consider a subfamily of self-adjoint extensions of $\wht{\phi}$ that
is small enough to allow the action of the gauge group to be written
down in an explicit form, yet broad enough to contain situations where
rigging maps of interesting structure can be extracted from the group
averaging formula~\eqref{eq:raqsesq}.

\subsection{Subfamily of classical rescalings and quantum 
boundary conditions}

We make two simplifying assumptions, one concerning the classical
rescaling function and the other concerning the quantum
self-adjointness boundary conditions.

First, we assume that \eqref{eq:sigma-inftyconditions}
fails for all $x$ for both signs, so that the formula 
\begin{align}
N(x) := 2\pi
\left(\int_{-\infty}^\infty 
\frac{\ud\tta'}{M(\tta',x)} \right)^{-1}
\end{align}
defines a function $N: \BbbR \to \BbbR_+$.
It follows that we can map 
$\haux$ to 
$\Ihaux := L^{2}(I\times \BbbR ,\ud\omega \, \ud x)$, 
where $I = [0,2\pi]$, 
by the Hilbert space isomorphism 
\begin{align}
\haux & \to \Ihaux , 
\notag
\\
\psi & \mapsto \psi_c , 
\notag
\\
\psi_c (\omega,x) & := 
\sqrt{\frac{M \bigl({\tilde\sgm}_{x}^{-1}(\omega/N(x)),x\bigr)}{N(x)}} 
\, \psi\bigl({\tilde\sgm}_{x}^{-1}(\omega/N(x)),x\bigr) , 
\end{align}
where 
\begin{align}
\label{tsgm-of-theta}
\tilde\sgm_{x}(\tta):=
\int_{-\infty}^{\, \tta}\frac{\ud\tta'}{M(\tta',x)} . 
\end{align}
The auxiliary inner product in $\Ihaux$ reads 
\begin{align}
\Iaux{\psi_c}{\chi_c}:=\int_{I\times \BbbR}  \ud\omega \, \ud x \; 
\ovl{{\psi_c}(\omega,x)} \, \chi_c(\omega,x) , 
\end{align}
and $\wht\phi$ \eqref{eq:phihat-def}
is mapped to 
\begin{align}
\label{eq:phihatc-def}
\wht{\phi}_c 
& := -i N(x) \, \partial_{\omega} . 
\end{align}
We work from now on in~$\Ihaux$, dropping the 
subscript $c$ from the wave functions. 

Second, we consider only those self-adjoint extensions of 
$\wht{\phi}_c$ \eqref{eq:phihatc-def} where the 
boundary conditions at $\omega=0$ and $\omega=2\pi$ 
do not couple different values of~$x$. 
The self-adjointness analysis then reduces 
to that of the momentum operator on an interval\citte{reebk2}, 
independently at each~$x$. 
The domains of self-adjointness are 
\begin{align}
\label{alfa-sadd}
D_{\al}:=\left\lbrace\psi,\partial_{\omega}\psi\in\Ihaux \mid \psi(\cdot,x) 
\in\text{ac}(0,2\pi)\  
\text{and}\ \psi(0,x)=e^{i2\pi\al(x)}\psi(2\pi,x),\ \forall x\right\rbrace,
\end{align}
where $\text{ac}(0,2\pi)$ denotes absolutely continuous functions of
$\omega$ and the function $\al:\BbbR\to\BbbR$ specifies the phase shift
between $\omega=0$ and $\omega=2\pi$ at each~$x$.

Under these assumptions, the remaining freedom in the 
classical rescaling function $M: \BbbR^2 \to \BbbR^2$ is 
encoded in the function $N: \BbbR \to \BbbR_+$, 
while the remaining freedom in the 
self-adjoint extension
of $\wht{\phi}_c$ \eqref{eq:phihatc-def} 
is encoded in the function $\al:\BbbR\to\BbbR$. 
Note that no smoothness assumptions about either
function are needed at this stage. 

The action of ${U}_c(\lm) := \exp\bigl(i\lm\wht{\phi}_c\,\bigr)$
takes now a simple form in a Fourier decomposition adapted
to~$D_{\al}$. We write each $\psi\in\Ihaux$ in the unique decomposition
\begin{align}
\label{eq:psi-decomp}
\psi(\omega,x) = 
\frac{1}{\sqrt{2\pi}}\sum_{n\in\BbbZ} 
e^{i\left[n-\al(x)\right]\omega} \, \psi_n (x) , 
\end{align}
where each $\psi_n$ is in $L^2(\BbbR, \ud x)$. It follows that 
\begin{align}
\label{eq:Iaux-decomp}
\Iaux{\psi}{\chi} = \sum_{n\in\BbbZ}  \Raux{\psi_n}{\chi_n}, 
\end{align}
where $\Raux{\cdot}{\cdot}$ is the inner product in $L^2(\BbbR, \ud x)$. 
The action of 
${U}_c(\lm)$ reads 
\begin{align}
\label{eq:Uc-on-psi}
\bigl({U}_c(\lm)\psi\bigr)_n (x) = 
e^{i R_n(x) \lm} \, \psi_n(x) , 
\end{align}
where for each $n\in\BbbZ$ the function $R_n: \BbbR \to \BbbR$ is defined by 
\begin{align}
\label{eq:Rq-def}
R_n(x) = \left[n-\al(x)\right] N(x) . 
\end{align}


\subsection{Test space, observables and rigging map candidates}

Let $\Phi$ be the dense linear subspace of $\Ihaux$ where 
the states have the form 
\eqref{eq:psi-decomp} such that 
every $\psi_n$ is smooth with 
compact support and only finitely many of them are 
nonzero for each $\psi\in\Phi$. 
From \eqref{eq:Uc-on-psi} we see that $\Phi$ 
is invariant under ${U}_c(\lm)$ for each~$\lm$. 
We adopt $\Phi$ as the RAQ test space  
of `sufficiently well-behaved' auxiliary states. 

Given $\Ihaux$, $\Phi$ and ${U}_c(\lm)$, the 
RAQ observables are operators $A$ on $\Ihaux$ such that the
domains of $A$ and $A^\dag$ include~$\Phi$,
$A$ and $A^\dag$ map $\Phi$ to itself
and $A$ commutes with ${U}_c(\lm)$ on $\Phi$ for all~$\lm$. 
We denote the algebra of the observables by~$\Aobs$.  

The final ingredient in RAQ is to specify the rigging map 
$\eta: \Phi \to \Phi^*$, 
where the star denotes the algebraic
dual, topologised by pointwise convergence. 
$\eta$~is antilinear, 
it must be real and positive in the sense that the properties 
\begin{subequations} 
\begin{align}
& \eta(f)[g] = \overline{\eta(g)[f]} \ , 
\\
& \eta(f)[f] \ge 0 \ , 
\label{eq:positivity}
\end{align}
\end{subequations}
hold for all $f, g \in \Phi$, 
and states in the image of $\eta$ must 
be invariant under the dual action of~${U}_c(\lm)$. 
Finally, $\eta$ must
intertwine with the representations of 
$\Aobs$ on $\Phi$ and $\Phi^*$
in the sense that 
\begin{align}
\label{eq:eta-intertwining-matrix}
\eta(A f) [g]
= 
\eta(f) [A^\dag g]
\ , 
\end{align}
for all $A\in\Aobs$ and $f, \, g \in \Phi$, 
where 
the left-hand side denotes the
dual action of $\eta(A f)\in\Phi^*$ 
on $g\in\Phi$ and 
the right-hand side denotes the dual action of 
$\eta(f) \in\Phi^*$
on $A^\dag g \in \Phi$. 
The physical Hilbert space 
$\Hraq$ is then the completion of the image of 
$\eta$ in the inner product 
\begin{align}
\ipraq{\eta(g)}{\vphantom{\bigm{|}}\eta(f)}
:= \eta(f)[g]
\ , 
\label{phys-ip}
\end{align}
and the properties of $\eta$ and $\Aobs$ imply that $\eta$ induces an
antilinear representation of $\Aobs$ on~$\Hraq$, 
with the image of
$\eta$ as the dense domain\citte{ash95,giu99a}. 

We seek a rigging map in the form 
\begin{align}
\eta(f)[g]
& = 
\lim_{\cutoff\to\infty} 
\frac{1}{\rho(\cutoff)}
\int_{-\cutoff}^{\cutoff}
\ud\lm
\Iaux{f}{\vphantom{\bigm{|}}{U}_c(\lm)g} 
\notag
\\
&= 
\lim_{\cutoff\to\infty}
\frac{1}{\rho(\cutoff)}
\sum_{n\in\BbbZ} 
\int_{-\infty}^\infty
\ud x \, \overline{f_n(x)}{g_n(x)} 
\int_{-\cutoff}^{\cutoff} 
\ud\lm
\, e^{i R_n(x) \lm} , 
\label{eq:eta-ga}
\end{align}
where the last expression follows from 
\eqref{eq:psi-decomp} 
and \eqref{eq:Uc-on-psi} 
after interchanging sums and integrals, justified by the 
assumptions about~$\Phi$. 
The function $\rho: \BbbR_+ \to \BbbR_+$ 
has been included in order to seek 
a finite answer in cases where the 
limit would otherwise diverge. 

The existence of the limit in \eqref{eq:eta-ga} 
depends delicately on the 
zero sets and the stationary point sets of the 
functions~$R_n$. 
In subsections \ref{sec:genintervals} and 
\ref{sec:gensmooth}
we introduce conditions 
that make the limit controllable.

\subsection{\texorpdfstring{$N$}{} and 
\texorpdfstring{$\al$}{} smooth, \texorpdfstring{$\al$}{} 
with integer-valued intervals}
\label{sec:genintervals}

We assume that $\al$ and $N$ are smooth. 
What will play a central role are the 
integer value sets of $\al$
and the stationary point 
sets of the functions 
$\left\{ R_n \mid n\in\BbbZ \right\}$. 
To control the stationary point sets, 
we assume that $R_n$ satisfy the following technical condition: 
\begin{itemize}
\item[(a)]
The stationary point set of each $R_n$
is either empty 
or the union of 
at most countably many isolated points, at most countably many 
closed intervals and at most two closed half-lines, 
such that any compact subset of $\BbbR$ contains at 
most finitely many of the isolated points 
and at most finitely many of the finite intervals. 
\end{itemize}
To control the integer value set of $\al$, 
we assume in this subsection the following condition: 
\begin{itemize}
\item[(b)]
$\al$ takes an integer value on at least one interval. 
\end{itemize}
It follows from (b) that at least one $R_n$ takes 
the value zero on an interval. 
Note that (a) and (b) include the special case 
where $\al$ takes an integer value everywhere, 
and the very special case where this integer value is zero. 

The group averaging formula \eqref{eq:eta-ga} 
takes the form 
\begin{align}
\eta(f)[g]
&= 
\lim_{\cutoff\to\infty}
\frac{2L}{\rho(\cutoff)}
\sum_{n\in\BbbZ} 
\left(
\int_{J_n}
\ud x \, \overline{f_n(x)}{g_n(x)} 
+ 
\int_{\BbbR \setminus J_n}
\ud x \, \overline{f_n(x)}{g_n(x)} \, 
\frac{\sin\bigl[\cutoff R_n(x)\bigr]}
{LR_n(x)} 
\right) , 
\label{eq:eta-ga-infty}
\end{align}
where $J_n \subset \BbbR$ is the union of all open 
intervals contained in the zero set of~$R_n$, 
that is, in the solution set of $\al(x) = n$. 
Setting $\rho(\cutoff) = 2\cutoff$, 
the second term in \eqref{eq:eta-ga-infty} vanishes 
by dominated convergence, 
and from the first term we obtain the map 
$\eta_{\infty}: \Phi \to \Phi^*$, 
\begin{align}
\label{eq:etainfty-def}
\bigl(\eta_\infty ( f ) \bigr)[g]
=  
\sum_{n\in\BbbZ} 
\int_{J_n}
\ud x \, 
\overline{{f}_n (x)} \, {g}_n (x)
\  . 
\end{align}
We have the following theorem. 
\begin{theorem}
\label{thm:nongeneric}
$\eta_{\infty}$ is a rigging map, with a nontrivial image. 
\end{theorem}
\begin{proof} 
All the rigging map axioms except the 
intertwining property \eqref{eq:eta-intertwining-matrix} 
are immediate. 
We verify \eqref{eq:eta-intertwining-matrix} 
in Appendix~\ref{app:rigging}.\hfill$\blacksquare$
\end{proof}

Group averaging has thus yielded a genuine 
rigging map $\eta_\infty$ after a suitable renormalisation. 
The Hilbert space $\mathcal{H}_{\infty}$ is separable and carries a 
nontrivial representation of~$\Aobs$. 
Comparison of 
\eqref{eq:Iaux-decomp} and \eqref{eq:etainfty-def}
shows that $\mathcal{H}_{\infty}$ can be antilinearly embedded in 
$\Ihaux$ as a Hilbert subspace, 
such that $\eta_\infty$ extends into 
the (antilinear) projection 
$L^2(\BbbR, \ud x) \to L^2(J_n, \ud x)$ 
in each of the components in~\eqref{eq:psi-decomp}. 

Note that the function $N$ does not appear 
in $\eta_\infty$~\eqref{eq:etainfty-def}, 
and the discussion in Appendix \ref{app:rigging} 
shows that the representation of
$\Aobs$ on the image of $\eta_\infty$ does not depend on $N$ either. 
The quantum theory has turned out completely 
independent of the remaining freedom 
in the rescaling function, 
even though the rescaling function may vary 
nontrivially over the sets 
$J_n$ that contribute in~\eqref{eq:etainfty-def}. 

In the special case where $\al(x)=0$ for all~$x$,  
we have 
\begin{align}
\label{eq:eta-cylord}
\eta_{\infty}(f)[g] = \Raux{f_0}{g_0} . 
\end{align}
Embedding $\mathcal{H}_{\infty}$ antilinearly as a Hilbert subspace
of $\Ihaux$ as above, 
this means that $\eta_{\infty}$ extends into the (antilinear) 
projection to the $n=0$ sector of~$\Ihaux$. 
When $N$ is a constant function, 
$N(x) = N_0$ for all~$x$, 
we can recover this extension of $\eta_{\infty}$ directly,
without introducing a test space, by 
noticing that the quantum gauge group 
$\{{U}_c(\lm) \mid \lm\in\BbbR\} \simeq \text{U}(1)$ is then compact and 
taking the group averaging formula to read 
\begin{align}
\eta(f)[g]
& = 
\frac{N_0}{2\pi}
\int_{0}^{2\pi/N_0}
\ud\lm
\Iaux{f}{\vphantom{\bigm{|}}{U}_c(\lm)g} 
\ , 
\label{eq:eta-ga-uone}
\end{align}
where the integration is over 
$\text{U}(1)$ exactly once. 
However, if $N$ is not constant, this shortcut 
is not available because the quantum gauge group is 
then still isomorphic to~$\BbbR$ rather than~$\text{U}(1)$.

\subsection{\texorpdfstring{$N$}{} and 
\texorpdfstring{$\al$}{} 
smooth and generic}
\label{sec:gensmooth}

In subsection \ref{sec:genintervals} the quantum theory arose 
entirely from the integer value intervals of~$\al$. 
We now continue to assume that $\al$ and $N$ are smooth, 
the technical stationary point condition (a) holds 
and $\al$ takes an integer value somewhere, 
but we take the integer value set of $\al$ to consist of 
isolated points. We first replace condition (b)
by the following:  
\begin{itemize}
\item[(b${}'$)]
The integer value set of $\al$ is non-empty, 
at most countable and without accumulation points, and 
$\al$ has a nonvanishing derivative of some order 
at each integer value. 
\end{itemize}
Second, we introduce the following notation for the zeroes of~$R_n$. 
Let $p$ be the order of the lowest nonvanishing derivative of $\al$ 
(and hence also of~$R_n$) 
at a zero of~$R_n$. 
For odd~$p$, we write the zeroes as~$x_{pnj}$, where the
last index enumerates the solutions with given $p$ and~$n$. For
even~$p$, we write the zeroes as~$x_{p \epsilon nj}$, where
$\epsilon\in\{1,-1\}$ is the sign of the $p$th derivative of $\al$ 
and the last index enumerates the zeroes with given $p$, $\epsilon$ and~$n$.
Let $\mcc{P}$ be the value set of the first index
of the zeroes $\{x_{p nj}\}$ and $\{x_{p \epsilon nj}\}$. 
Given this notation, we assume: 
\begin{itemize}
\item[(c)]
If $p\in\mcc{P}$, then
$\mcc{P}$ contains no factors of 
$p$ smaller than~$p/2$. 
\end{itemize}

Before examining the group averaging formula \eqref{eq:eta-ga} under
these assumptions, we use the assumptions to define directly a family
of rigging maps as follows.  For each odd $p\in\mcc{P}$ we define the
map
$\eta_p: \Phi \to \Phi^*$, 
and for each 
even $p\in\mcc{P}$ and 
$\epsilon\in\{1,-1\}$ for which the set 
$\{x_{p \epsilon nj}\}$ is non-empty, we define the map 
$\eta_{p\epsilon}: \Phi \to \Phi^*$, by the formulas 
\begin{subequations}
\label{eq:eta-p-peps}
\begin{align}
\bigl(\eta_p ( f ) \bigr)[g]
&=  
\sum_{nj}
\frac{\overline{{f}_n (x_{pnj})} \, {g}_n (x_{pnj})}
{{\bigl|  \al^{(p)}(x_{pnj}) N(x_{pnj}) \bigr|}^{1/p}} 
\  , 
\\
\label{eq:eta-p-ftilde1}
\bigl(\eta_{p\epsilon} ( f ) \bigr) [g]
& =  
\sum_{nj}
\frac{\overline{{f}_n (x_{p \epsilon nj})} \, {g}_n (x_{p \epsilon nj})}
{{\bigl| \al^{(p)}(x_{p \epsilon nj}) N(x_{p \epsilon nj}) \bigr|}^{1/p}} 
\  . 
\end{align}
\end{subequations} 
These maps are rigging maps, 
with properties given in the following theorem.  
\begin{theorem}
\label{thm:generic}
$\phantom{xxx}$ 
\begin{enumerate}
\item 
Each $\eta_p$ and $\eta_{p\epsilon}$ 
is a rigging map, with a nontrivial image. 
\item 
The representation of $\Aobs$ on the image of each 
$\eta_p$ and $\eta_{p\epsilon}$ is irreducible. 
\end{enumerate}
\end{theorem}
\begin{proof} 
$\phantom{xxx}$
\begin{enumerate}
\item 
All the rigging map axioms except the 
intertwining property \eqref{eq:eta-intertwining-matrix} 
are immediate from~\eqref{eq:eta-p-peps}. 
We verify \eqref{eq:eta-intertwining-matrix} 
in Appendix~\ref{app:rigging}. 
\item 
The proof is an almost verbatim 
transcription of that 
given for a closely similar system 
in Appendix C of\citte{lou05}. 
We omit the details.\hfill$\blacksquare$
\end{enumerate}
\end{proof}

The rigging maps \eqref{eq:eta-p-peps} thus yield a 
family of quantum theories, 
one from each $\eta_p$ and~$\eta_{p\epsilon}$. 
Each of the Hilbert spaces is either finite-dimensional or 
separable and carries a nontrivial 
representation of $\Aobs$ that is 
irreducible on its dense domain. 
Functions $f\in \Phi$ whose only nonvanishing component 
$f_n$ is non-negative and is positive only near 
a single zero of $R_n$ provide 
the Hilbert spaces with a canonical orthonormal basis. 

Proceeding as in Appendix C of\citte{lou05}, 
we see that the representation of 
$\Aobs$ on the image of each 
$\eta_p$ and $\eta_{p\epsilon}$
is not just irreducible but has the following stronger property, 
which one might call strong irreducibility: 
given any two vectors $v$ and $v'$ 
in the canonical orthonormal basis, 
there exists an element of $\Aobs$ that 
annihilates all the basis vectors except 
$v$ and takes $v$ to~$v'$. 
The upshot of this is that 
the function $N$ plays little role in the quantum theory, 
despite appearing 
in the rigging map formulas~\eqref{eq:eta-p-peps}. 
The Hilbert spaces and their canonical 
bases are determined by the function $\al$
up to the normalisation of the individual basis vectors, 
and the representation of $\Aobs$ is so `large' 
that the normalisation of the individual basis vectors, 
determined by~$N$, is of limited consequence. 
In particular, 
the representation of $\Aobs$ on any Hilbert space with 
dimension $n_0<\infty$ is isomorphic to the 
complex $n_0 \times n_0$ matrix algebra, independently of~$N$. 

Now, we wish to relate these 
quantum theories to the 
group averaging formula~\eqref{eq:eta-ga}, 
which takes the form 
\begin{align}
\eta(f)[g]
&= 
\lim_{\cutoff\to\infty}
\frac{2}{\rho(\cutoff)}
\sum_{n\in\BbbZ} 
\int_{-\infty}^\infty
\ud x \, \overline{f_n(x)}{g_n(x)} \, 
\frac{\sin\bigl[\cutoff R_n(x) \bigr]}
{R_n(x)} \ .  
\label{eq:eta-ga-1}
\end{align}
Note that the integral over $x$ in \eqref{eq:eta-ga-1} 
is well defined because the zeroes of the denominator are isolated and the 
integrand does not diverge at them. 

Suppose first that $\mcc{P} = \{1\}$ and we set $\rho(\cutoff)=2\pi$. 
The lemmas of Appendix \ref{lemmas-app} then show that \eqref{eq:eta-ga-1}
is well defined and equals 
$\eta_1(f)[g]$ provided the assumptions on $N$ are modestly strengthened, 
in particular to preclude any $R_n$ from  
taking a constant value on any interval. 

Suppose then that $\mcc{P} \ne \{1\}$ and we again set
$\rho(\cutoff)=2\pi$. Suppose further that the assumptions on $N$ are
again modestly strengthened so that the conditions of Appendix
\ref{lemmas-app} hold, and suppose that condition (c) above is
strengthened to the following:
\begin{itemize}
\item[]
(c${}'$)\ \ 
If $p\in\mcc{P}$, then
$\mcc{P}$ contains no factors of~$p$. 
\end{itemize}
The lemmas of Appendix \ref{lemmas-app} then 
show that \eqref{eq:eta-ga-1}
contains contributions that diverge in the $\cutoff\to\infty$ limit; 
however, these divergences come in well-defined inverse 
fractional powers of $\cutoff$ such that the coefficient of 
each $\cutoff^{(p-1)/p}$ is proportional to $\eta_p(f)[g]$ for odd $p$ 
and to $\eta_{p,1}(f)[g] + \eta_{p,-1}(f)[g]$ for even~$p$. 

When $\mcc{P} = \{1\}$, we may hence regard the rigging map $\eta_1$
as arising from \eqref{eq:eta-ga} with only minor strengthening of our
technical assumptions.  When $\mcc{P} \ne \{1\}$, we may regard the
rigging maps $\eta_p$ and $\eta_{p,1} + \eta_{p,-1}$ as arising from
\eqref{eq:eta-ga} by peeling off and appropriately renormalising the
various divergent contributions, but only after strengthening the
assumptions so that some generality is lost, and even then the two
signs of $\epsilon$ are recovered only in a fixed linear combination
but not individually.

We end with two technical comments. 
First, it may be possible to find assumptions 
that interpolate between those in subsections 
\ref{sec:genintervals} and~\ref{sec:gensmooth}, 
allowing both a superselection sector that 
comes from integer-valued intervals of 
$\alpha$ and superselection sectors that come 
from isolated zeroes of~$\alpha$. 
In formula~\eqref{eq:eta-ga-infty}, 
the task would be to provide a peeling argument 
in the $\cutoff$-dependence of the second term. 
In the 
observable analysis of Appendix~\ref{app:rigging}, 
the task would be to provide 
a peeling argument in the small 
$|s|$ behaviour of the 
integrands in~\eqref{eq:ifourier-genuine}. 

Second, 
our quantum theories arise from the integer value set of~$\al$, 
both in subsection \ref{sec:genintervals} 
and in subsection~\ref{sec:gensmooth}. 
Neither the averaging formulas nor the observable analysis of 
Appendix \ref{app:rigging} suggest ways to proceed when 
$\al$ takes no integer values. 
In~\eqref{eq:eta-ga-1}, 
the challenge would be to recover from the oscillatory 
$\cutoff$-dependence a map that satisfies 
the positivity condition~\eqref{eq:positivity}. 
A similar oscillatory dependence on $\lm$ 
presents the challenge in the observable formula~\eqref{eq:ifourier}.

\section{Summary and discussion}
\label{remarks}

In this paper we have investigated refined algebraic quantisation
under rescalings of a single momentum-type constraint in a Hamiltonian
system whose unreduced configuration space is~$\BbbR^2$.  While such
rescalings do not affect the classical reduced phase space, they do
affect the options to find a rigging map by which the constraint is
implemented in the quantum theory. We found that the rescalings fall
into three cases, depending on the choice of the rescaling function. In
case~(i), the rescaled constraint operator is essentially self-adjoint
in the auxiliary Hilbert space, and the quantisation is equivalent to
that with identity scaling. In case~(ii), the rescaled constraint
operator has no self-adjoint extensions and no quantum theory is
recovered.  In case~(iii), the rescaled constraint operator admits a
family of self-adoint extensions, and the choice of the extension has
a significant effect on the quantum theory.  In particular, the choice
determines whether the quantum theory has superselection sectors.

Within case~(iii), we analysed in full a subfamily of rescalings and
self-adjoint extensions in which the superselection structures turned
out to resemble closely that of the Ashtekar-Horowitz-Boulware
model\citte{lou05}. There are however two significant differences, one
conceptual and one technical.  Conceptually, the superselection
sectors in the Ashtekar-Horowitz-Boulware model are determined by the
\emph{classical\/} potential function in the constraint, while in our
system the superselection sectors are determined by a quantisation
ambiguity that has no counterpart in the classical system.
Technically, in our system it is `natural' to consider a wider family
of self-adjoint extensions than the family of potential functions
considered in\citte{lou05}, and we duly found a wider set of quantum
theories.  In particular, while all the quantum theories in
\citte{lou05} have finite-dimensional Hilbert spaces, some of our
quantum theories have separable Hilbert spaces, and some of them can
even be realised as genuine Hilbert subspaces of the auxiliary Hilbert
space.

Within those case (iii) theories that we analysed in full, we found
the quantum theory to be insensitive to the remaining freedom in the
rescaling function.  We in particular discovered situations where the
quantum gauge group is $\BbbR$ for generic rescaling functions but
reduces to $\text{U}(1)$ in the special case of a constant rescaling
function: yet this difference between a compact and noncompact gauge
group was irrelevant for the quantum theory, and the quantum theory
coincided with that which is obtained with the compact gauge group by
a projection into the $\text{U}(1)$-invariant subspace of the
auxiliary Hilbert space. 
The formalism of refined algebraic quantisation is thus here
able to handle seamlessly the transition between a compact and a
noncompact gauge group.

As our system has just one constraint, the quantum gauge
transformations form an Abelian Lie group both before and after the
constraint rescaling.  In a system with more constraints, constraint
rescalings can relate closed gauge algebras to open ones, and even
among closed algebras they can change the underlying Lie group, in
particular taking an Abelian Lie group to a non-Abelian one. Extending
the analysis of this paper to more than one constraint via the BRST
tools of\citte{shv02} would hence raise a number of new
issues. However, we emphasise that while the search for rigging maps
in this paper used group averaging as the starting point, the
nontrivial part in showing that a rigging map is actually recovered
was in the action of the quantum gauge transformations on the
observables, and in subsection \ref{sec:gensmooth} a direct analysis
of these observables allowed us in fact to find more rigging maps than
suggested by the group averaging.  Should notions of averaging be
difficult to generalise to rescalings with more than one constraint,
it may hence well be sufficient to focus directly on the action of the
quantum gauge transformations on the observables.

\section*{Acknowledgments}

We thank 
Konstantin Ardakov, 
Jurek Lewandowski 
and 
Joachim Zacharias 
for helpful discussions and correspondence. 
J.~L. was supported in part by STFC~(UK). 
E.~M.--P. was supported by CONACYT (Mexico).

\appendix

\section{Appendix: BRST quantisation and the averaging proposal}
\label{brst-analysis}

In this appendix we review how the group averaging formula
\eqref{eq:raqsesq} for
our system emerges from the BRST formalism, adopting the conventions
of\citte{shv02}. For detailed expositions of the BRST formalism we
refer to\citte{henbk,henrep}.

Domains of operators are unspecified throughout the appendix and
Hermiticity considerations remain formal.


\subsection{Classical BRST formalism}

Let $q$ and $p$ denote respectively the coordinates $(\tta,x)$ 
and the momenta $(p_\tta,p_x)$ 
on the original phase space~$\Gm$. 
The new canonical variables are 
the Lagrange multiplier~$\lm$, 
the ghost~$\c$, 
the antighost $\bc$ 
and their respective conjugate momenta 
$\pi$, $\bp$ and~$\p$. 
The ghost number $\ghalt{\cdot}$ and Grassmann parity $\eps(\cdot)$ 
of the variables are
\begin{subequations}
\begin{align}
\ghalt{q}=\ghalt{\lm}=\ghalt{p}=\ghalt{\pi}=0, & 
\ \ \ \eps(q)=\eps(\lm)=\eps(p)=\eps(\pi)=0,
\label{gngp-qp}
\\
\ghalt{\c}=\ghalt{\p}  =1, & \ \ \ \eps(\c)=\eps(\p)=1,\ 
\label{gngp-cp}
\\
\ghalt{\bc}=\ghalt{\bp}  =-1, & \ \ \ \eps(\bc)=\eps(\bp)=1 . 
\label{gngp-bcp}
\end{align}
\end{subequations}
All the bosonic variables are real-valued. 
Of the fermionic variables, we take the pair 
$(\c,\bp)$ to be real and the pair $(\bc,\p)$ 
purely imaginary\citte{shv02}. 
The nonvanishing (graded) 
Poisson brackets are 
\begin{subequations}
\label{sseps}
\begin{align}
&\{\tta,\ p_{\tta}\} = \{x,\ p_{x}\} = \{\lm,\ \pi\} =1,
& \text{(bosonic)}&\\
&\bigl\{\c,\ \bp \, \bigr\} =  \bigl\{\bc,\ \p\bigr\} = -i . & 
\text{(fermionic)}  
\label{sseps:fermionic}
\end{align}
\end{subequations}

We note in passing that the 
fermionic brackets \eqref{sseps:fermionic} 
are imaginary because of the fermionic 
reality conditions. If $\c$ and $\bc$ are instead 
chosen real and their conjugate momenta imaginary\citte{henbk}, 
the fermionic brackets \eqref{sseps:fermionic} must be taken real, 
with concomitant changes in the subsequent formulas; 
in particular, a Hermitian $\brs{\cdot}{\cdot}$ 
is then obtained by $c=\pm i$
in \eqref{aibrsp} below. 
The fermionic reality convention  
does however 
not affect the content of the resulting quantum theory. 

As the original Lagrange multiplier $\lm$ 
has become a phase space variable, 
the extended system has two constraints: 
the original constraint $\phi$ \eqref{rc} 
and the new constraint~$\pi$. The BRST generator $\Omega$ 
has contributions from both
constraints and reads
\begin{align}
\label{brsch}
\Om := \phi \,\c -i \pi \p . 
\end{align}
$\Om$ is real and
satisfies $\{\Om,\Om\}=0$. 

\subsection{BRST quantisation}

We choose a representation in which 
the wave functions depend on the bosonic coordinates 
$(\tta,x,\lm)$ and the fermionic momenta $(\bp,\p)$. 
A~wave function can be expanded in the fermionic variables as 
\begin{align}
\label{brswf}
\Psi(\tta,x,\lm,\bp,\p) = \psi(\tta,x,\lm)
+\Psi^{1}(\tta,x,\lm)\bp+\Psi_{1}(\tta,x,\lm)\p
+\Psi^{1}_{1}(\tta,x,\lm)\bp\p, 
\end{align}
where $\psi$, $\Psi^1$, $\Psi_1$ and $\Psi_1^1$ 
are complex-valued. 
The action of the 
fundamental operators reads 
\begin{subequations} 
\label{basop}
\begin{align}
&&&\hat{\tta}\Psi  := \tta\Psi, 
&&\hat{p}_{\tta}\Psi:=-i\dnd{}{\Psi}{\tta}, 
&&\\
&&&\hat{x}\Psi :={x}\Psi, 
&&\hat{p}_{x}\Psi:=-i \dnd{}{\Psi}{x}, 
&&\\
&&&\hat{\lm}\Psi  := \lm\Psi, 
&&\hat{\pi}\Psi:=-i\frac{\partial\Psi}{\partial\lm}, 
&&\\
&&&\hat{\c}\Psi  := \dndl{\Psi}{\bp}, 
&&\hat{\bp}\Psi:=\bp\Psi, 
&&\\
&&&\hat{\bc}\Psi:= \dndl{\Psi}{\p}, 
&&\hat{\p}\Psi := \p\Psi, 
&&
\end{align} 
\end{subequations} 
where the superscript $l$ on the fermionic derivative indicates the
left derivative.  The (graded) commutators of the fundamental
operators are equal to $i$ times the corresponding (graded) Poisson
brackets~\eqref{sseps}:
\begin{subequations}
\label{q-sseps}
\begin{align}
&\bigl[\hat\tta,\ \hat p_{\tta}\bigr] 
= 
\bigl[\hat x,\ \hat p_{x}\bigr] 
= 
\bigl[\hat\lm,\ \hat\pi\bigr] =i,
& \text{(bosonic)}&\\
&
\bigl[\hat\c,\ \hat\bp \, \bigr] 
=  
\bigl[ \, \hat\bc,\ \hat\p\bigr] = 1. & 
\text{(fermionic)}  
\end{align}
\end{subequations}

The physical quantum states satisfy 
\begin{subequations}
\label{eq:brsc-and-ghnopc}
\begin{align}
\label{brsc}
\wht{\Om} \Psi &=0, 
\\
\label{eq:ghnopc}
\wht{N}_{\mc{G}}\Psi &=0,  
\end{align}
\end{subequations}
where the BRST operator $\wht{\Om}$ and the ghost number operator
$\wht{N}_{\mc{G}}$ are defined by
\begin{align}
\label{qbrso}
\wht{\Om} & :=
\wht{\phi}\,\hat{\c}-i \hat{\pi} \hat{\p} , 
\\
\label{eq:phihat-def}
\wht{\phi} & :=-i\left(M\partial_{\tta}
+\tfrac{1}{2}(\partial_{\tta}M)\right) , 
\\
\label{gnop}
\wht{N}_{\mc{G}}
& :=
\hat{\p}\hat{\bc}-\hat{\bp}\hat{\c} . 
\end{align}
If $X$ is any state, the transformation 
\begin{align}
\Psi \mapsto \Psi' := \Psi + \wht{\Om} X
\end{align}
is called a BRST gauge transformation, and states related by a gauge
transformation are called gauge-equivalent. 
As $\bigl[\wht{\Om},\wht{\Om}\bigr]=
2 (\wht{\Om}){\vphantom{\bigl(\bigr)}}^{2}=0$, a gauge transformation
preserves the condition~\eqref{brsc}, and if $X$ has ghost
number~$-1$, a gauge transformation also preserves the
condition~\eqref{eq:ghnopc}.  A~gauge transformation in which $X$ has
ghost number $-1$ hence takes physical states to physical states.

The BRST `inner product' is the sesquilinear form 
\begin{align}
\label{aibrsp}
\brs{\Psi}{\Upsilon}
:=
c \int 
\ud\lm\,\ud\tta\,\ud x
\, 
\ud\bp\,\ud\p
\ 
\Psi^{*}(\tta,x,\lm,\bp,\p)
\Ups(\tta,x,\lm,\bp,\p) , 
\end{align}
where ${}^*$ denotes complex conjugation and $c$ is a nonzero 
constant that may a priori take complex values. This
definition has a number of desirable properties that are independent
of~$c$.  First, $\brs{\cdot}{\cdot}$ is compatible with the reality
conditions of the classical fundamental variables, in the sense that
$\hat\bc$ and $\hat\p$ are antihermitian and all the other fundamental
operators in \eqref{basop} are Hermitian.  Second, the BRST operator
$\wht{\Om}$ is Hermitian with respect to $\brs{\cdot}{\cdot}$, which 
property is compatible with the reality of the classical BRST
charge~$\Om$: the only nontrivial ordering issue in $\wht{\Om}$ is
that of the purely bosonic factor $\wht{\phi}$~\eqref{eq:phihat-def}.
Third, from the Hermiticity of $\wht{\Om}$ it follows that
$\brs{\cdot}{\cdot}$ on physical states depends on the states only
through their gauge-equivalence class.

If $c$ is real, $\brs{\cdot}{\cdot}$ is Hermitian,
but it fails to provide a genuine inner product because it is not
positive definite. We shall comment on the choice of $c$ below.

\subsection{Averaging}

To connect the BRST quantisation to a formalism that only involves
bosonic variables, it is not possible simply to drop all powers of the
fermions from the quantum states since the fermionic integrations in
\eqref{aibrsp} annihilate such states. There is however the option to
start from states without fermions and evaluate $\brs{\cdot}{\cdot}$
on suitable gauge-equivalent states. 

Suppose that 
$\Psi$ and $\Upsilon$ are physical states without fermions. 
The physical state conditions \eqref{eq:brsc-and-ghnopc} 
imply that the states take the form 
\begin{align}
\label{eq:purebos-states}
\Psi=\psi(\tta,x), \ \ \ \Upsilon=\chi(\tta,x), 
\end{align}
where the $\lm$-independence follows from the BRST
condition~\eqref{brsc}. We wish to define a regularised sesquilinear
form $\brsren{\cdot}{\cdot}$ by
\begin{align}
\label{bmip-r}
\brsren{\psi}{\chi} := \brs{\psi}{\wht{V}\chi} , 
\end{align}
where 
$\wht{V}:=\exp\bigl([\wht{\Om},\wht{K}]\bigr)$ and 
$\wht{K}$ is a suitable operator with ghost number~$-1$. 
Note that as $\chi$ and $\wht{V}\chi$ are 
gauge-equivalent physical states, 
the right-hand side of \eqref{bmip-r} would be 
independent of $\wht{K}$ if well defined for all~$\wht{K}$. 
$\wht{K}$~is called the gauge-fixing fermion. 

The usual procedure is to choose a Hermitian
gauge-fixing fermion by 
$\wht{K} =-\hat{\lm}\hat{\bp}$\citte{henbk,shv02,Marnelius:1990eq,%
Marnelius:1993ba,bat95,mar95}. 
It follows that
$[\wht{\Om},\wht{K}] = -\lm\wht{\phi} - \bp\p$.  
The integrations over the ghost momenta in 
\eqref{bmip-r} are 
elementary and we obtain 
\begin{align}
\label{bmip-r-K-herm}
\brsren{\psi}{\chi} = 
c 
\int \ud\lm \, \ud\tta \, \ud x
\, 
\psi^{*}(\tta,x)\bigl[\exp\bigl(-\lm\wht{\phi}\,\bigr)\chi\bigr](\tta,x) . 
\end{align}
The constant $c$ is then chosen equal to~$1$. Finally, the
quantisation of the pair $(\lm, \pi)$ is understood in a sense that
makes the spectrum of $\hat\lm$ purely
imaginary\citte{pauli-indef}. The final formula for
$\brsren{\cdot}{\cdot}$ reads
\begin{align}
\label{bmip-r-mu}
\brsren{\psi}{\chi} = 
\int \ud\mu \, \ud\tta \, \ud x
\, 
\psi^{*}(\tta,x)\bigl[\exp\bigl(i\mu\wht{\phi}\,\bigr)\chi\bigr](\tta,x) , 
\end{align}
where $\mu$ is real-valued. Formula \eqref{bmip-r-mu} provides the
candidate for a refined algebraic quantisation sesquilinear form for
the purely bosonic system, and it is our starting point
\eqref{eq:raqsesq} in section~\ref{generic-M}.

An alternative is to choose the antihermitian
gauge-fixing fermion $\wht{K}
=i\hat{\lm}\hat{\bp}$\citte{Marnelius:1990eq}. 
This choice makes $\wht{V}$ unitary, and
integration over the ghosts yields
\begin{align}
\label{bmip-r-K-antiherm}
\brsren{\psi}{\chi} = 
-i c 
\int \ud\lm \, \ud\tta \, \ud x
\, 
\psi^{*}(\tta,x)\bigl[\exp\bigl(i\lm\wht{\phi}\,\bigr)\chi\bigr](\tta,x) . 
\end{align}
Choosing now $c=i$ and quantising the pair $(\lm, \pi)$ 
in a way that keeps the spectrum of $\hat\lm$ real, 
we again arrive at~\eqref{bmip-r-mu}.

\section{Appendix: Intertwining property of the rigging maps}
\label{app:rigging}

In this Appendix we verify
that the rigging maps \eqref{eq:etainfty-def} and \eqref{eq:eta-p-peps}
have the intertwining property~(\ref{eq:eta-intertwining-matrix}), 
completing the proof of 
Theorems \ref{thm:nongeneric} and~\ref{thm:generic}\null. 
We follow the method introduced 
in Appendix B of\citte{lou05}.  

To begin, we assume just that $\al$ and $N$ satisfy 
condition (a) of subsection~\ref{sec:genintervals}. 
The fork between the remaining conditions 
of subsections \ref{sec:genintervals} and \ref{sec:gensmooth}
takes place after~\eqref{eq:ifourier-both}. 

Let $A\in\Aobs$. Let $m$ and $n$ be fixed integers and let 
$f, \, g \in
\Phi$ such that their only components in the decomposition 
\eqref{eq:psi-decomp}
are respectively $f_m$ and~$g_n$. 
As $U_c(\lm)$ is unitary and commutes with~$A^\dag$,
we have 
$\Iaux{U_c(-\lm)f}{A^\dag g}
= 
\Iaux{f}{U_c(\lm) A^\dag g}
= 
\Iaux{f}{A^\dag U_c(\lm) g}
=
\Iaux{A f}{U_c(\lm) g}$. 
Using \eqref{eq:Iaux-decomp} and~\eqref{eq:Uc-on-psi}, 
the leftmost and rightmost expressions yield  
\begin{align}
\int\! \ud x\, 
e^{i R_m(x)\lm } 
\, 
\overline{f_m(x)} \, \bigl(A^\dag g\bigr)_m(x) 
=
\int\! \ud x\,
e^{i R_n(x) \lm } 
\, 
\overline{\bigl(A f\bigr)_n(x)} 
\, g_n (x)
\ . 
\label{eq:iidentity}
\end{align}

We denote the intervals in which 
$R_q$ has no stationary points by~$I_{qr}$, 
where the second index $r$ enumerates the intervals with given~$q$. 
We similarly denote the intervals in which $R_q$ is constant by~$\tilde I_{q\tilde r}$. 
We take these intervals to be open and inextendible, 
and we understand ``interval'' to include half-infinite 
intervals and the full real line. 

On the left-hand side (respectively right-hand side) 
of~\eqref{eq:iidentity}, 
we break the integral over $x\in \BbbR$ into a sum 
of integrals over $\{I_{m r}\}$
and $\{\tilde I_{m\tilde r}\}$ 
($\{I_{n r}\}$
and~$\{\tilde I_{n\tilde r}\}$). 
By condition (a) of subsection \ref{sec:genintervals} 
and the assumptions about~$\Phi$, 
the sums contain at most finitely many terms.

Let $R_{q r}$ be the
restriction of $R_q$ to~$I_{q r}$, 
and let $R_{q r}^{-1}$ be the
inverse of~$R_{q r}$. 
Changing the integration variable in each
$I_{m r}$ on the left-hand side to 
$s:= R_{m r}(x)$ and 
in each 
$I_{n r}$ on the right-hand side to 
$s:= R_{n r}(x)$, we obtain 
\begin{align}
& \phantom{=}\, 
\sum_{\tilde r}
\int_{\tilde I_{m\tilde r}}\! 
\ud x\, 
e^{i R_m(x)\lm } 
\, 
\overline{f_m(x)} \, \bigl(A^\dag g\bigr)_m(x) 
+
\int\! \ud s\, 
e^{i\lm s } 
\sum_r
\left[ 
\frac{\overline{f_m} \, \bigl(A^\dag g\bigr)_m}
{|R_m'|}
\right]
\bigl( R_{mr}^{-1} (s) \bigr)
\notag
\\
& =
\sum_{\tilde r}
\int_{\tilde I_{n\tilde r}}\! 
\ud x\, 
e^{i R_n(x) \lm } 
\, 
\overline{\bigl(A f\bigr)_n(x)} 
\, g_n (x)
+ 
\int\! \ud s\,
e^{i\lm s} 
\sum_r
\left[ 
\frac{\overline{\bigl(A f\bigr)_n} \, g_n}
{|R_n'|}
\right]
\bigl( R_{n r}^{-1} (s) \bigr)
\ , 
\label{eq:ifourier}
\end{align}
where for given $s$ the sum over $r$ 
on the left-hand side (respectively right-hand side) is
over those $r$ for which $s$ 
is in the image of $R_{mr}$ 
($R_{nr}$). 

We now regard each side of \eqref{eq:ifourier} 
as a function of $\lm\in\BbbR$. On each side, 
the integral over $s$ 
is the Fourier transform of an $L^1$ function and hence
vanishes as $|\lm|\to\infty$ by the Riemann-Lebesgue lemma, 
whereas the sum over $\tilde r$ is a finite linear combination of 
imaginary exponentials and does not vanish as 
$|\lm|\to\infty$ unless identically zero. 
A~peeling argument shows that \eqref{eq:ifourier} breaks into the pair 
\begin{subequations}
\label{eq:ifourier-both}
\begin{align}
\sum_{\tilde r}
\int_{\tilde I_{m\tilde r}}\! 
\ud x\, 
e^{i R_m(x)\lm } 
\, 
\overline{f_m(x)} \, \bigl(A^\dag g\bigr)_m(x) 
& =
\sum_{\tilde r}
\int_{\tilde I_{n\tilde r}}\! 
\ud x\, 
e^{i R_n(x) \lm } 
\, 
\overline{\bigl(A f\bigr)_n(x)} 
\, g_n (x) \ , 
\label{eq:ifourier-almostper}
\\
\int\! \ud s\, 
e^{i\lm s } 
\sum_r
\left[ 
\frac{\overline{f_m} \, \bigl(A^\dag g\bigr)_m}
{|R_m'|}
\right]
\bigl( R_{mr}^{-1} (s) \bigr)
& =
\int\! \ud s\,
e^{i\lm s} 
\sum_r
\left[ 
\frac{\overline{\bigl(A f\bigr)_n} \, g_n}
{|R_n'|}
\right]
\bigl( R_{nr}^{-1} (s) \bigr)
\ . 
\label{eq:ifourier-genuine}
\end{align}
\end{subequations}

Suppose now that condition (b) of subsection 
\ref{sec:genintervals} holds. A~peeling argument 
shows that the $\lm$-independent component of 
\eqref{eq:ifourier-almostper} reads 
\begin{align}
\eta_\infty(f) [A^\dag g]
=
\eta_\infty(A f) [g]
\ , 
\label{eq:intert-infty} 
\end{align}
where $\eta_\infty$ is defined in~\eqref{eq:etainfty-def}. 
By linearity, \eqref{eq:intert-infty} continues to hold for all $f$ and $g$ in~$\Phi$. 
$\eta_\infty$~hence has the intertwining property~(\ref{eq:eta-intertwining-matrix}). 
This completes the proof of Theorem~\ref{thm:nongeneric}. 

Suppose then that conditions (b${}'$) and (c) 
of subsection \ref{sec:gensmooth} hold. 
Examination of the integrands in 
\eqref{eq:ifourier-genuine} near $s=0$ 
by the technique of 
Appendix B of\citte{lou05} shows that 
\begin{subequations}
\label{eq:intert-dd-fg}
\begin{align}
\eta_p(f) [A^\dag g]
&=
\eta_p(A f) [g]
\ , 
\\
\eta_{p\epsilon}(f) [A^\dag g]
&= 
\eta_{p\epsilon}(A f) [g]
\ , 
\end{align}
\end{subequations}
for all $p$ and $\epsilon$ for which the maps 
$\eta_p$ and 
$\eta_{p\epsilon}$
\eqref{eq:eta-p-peps} are defined. By linearity,  
\eqref{eq:intert-dd-fg} continues to hold for all $f$ and $g$ in~$\Phi$. 
Each $\eta_p$ and 
$\eta_{p\epsilon}$ hence has the intertwining
property~(\ref{eq:eta-intertwining-matrix}). 
This completes the proof of Theorem~\ref{thm:generic}.

\section{Appendix: Lemmas on asymptotics}
\label{lemmas-app}

In this appendix we record two lemmas on 
asymptotics of 
integrals that occur in section~\ref{case-iii}. 

\begin{lemma}
\label{lemm:asym1}
Let $f \in C^\infty_0(\BbbR)$, $\cutoff>0$, $p \in \{1,2,\ldots\}$ and  
\begin{align}
\label{eq:Gp-def}
G_p(\cutoff) := \int_{-\infty}^{\infty}
f(u) 
\, 
\frac{\sin(\cutoff u^p)}{u^p}
\, \ud u . 
\end{align}
As $\cutoff \to\infty$, 
\begin{align}
\label{eq:blemmaeq}
G_p(\cutoff) = 
\sum_{q=0}^{p-1}
K_{p,q} \, f^{(q)}(0)
\, 
\cutoff^{(p-1-q)/p}
\ \ 
+ O\bigl(\cutoff^{-1/p}\bigr)
\end{align}
where 
\begin{align}
K_{p,q} &= 
\frac{\sqrt{\pi} 
\, 2^{(q+1-p)/p}\Gamma\bigl(\frac{q+1}{2p}\bigr)}
{p \, q! \, \Gamma\bigl(\frac{3p-q-1}{2p}\bigr)} . 
\end{align}
\end{lemma}
\begin{proof} 
(Sketch.)
We replace $f(u)$ in \eqref{eq:Gp-def} by its Taylor 
series about the origin, 
including terms up to~$u^{p-1}$, 
at the expense of an error 
that is $O\bigl(\cutoff^{-1/p}\bigr)$. 
The terms in the Taylor series give respectively 
the terms shown in 
\eqref{eq:blemmaeq} plus an 
error that is $O\bigl(\cutoff^{-1}\bigr)$.\hfill$\blacksquare$
\end{proof}

Let $f\in C^\infty_0(\BbbR)$ and $R \in C^\infty(\BbbR)$. 
Let $R$ have at most finitely many zeroes and at most 
finitely many stationary points, and let all 
stationary points of $R$ be of finite order. 
Denote the zeroes of $R$ by~$x_{pj}$, 
where $p \in \{1,2,\ldots\}$ is the order of the lowest nonvanishing 
derivative of $R$ at $x_{pj}$ and $j$ enumerates the zeroes with given~$p$. 
For $\cutoff>0$, let 
\begin{align}
I(\cutoff) := \int_{-\infty}^{\infty}
f(x) 
\, 
\frac{\sin\bigl[\cutoff R(x)\bigr]}{R(x)}
\, dx \ . 
\end{align}

\begin{lemma}
\label{lemm:asym2}
As $\cutoff\to\infty$, 
\begin{align}
\label{lemma3}
I(\cutoff) = \sum_{pj} I_{pj}(\cutoff)  
\ \ 
+ o(1)
\end{align} 
where 
\begin{align}
\label{lemma2}
I_{pj}(\cutoff) = 
K_{p,0} \, 
{\left(\frac{p!}{\bigl|R^{(p)}(x_{pj})\bigr|}\right)}^{\!1/p}
f(x_{pj})
\, 
\cutoff^{(p-1)/p}
\ + \ 
\sum_{q=1}^{p-1} A_{pjq} \cutoff^{(p-1-q)/p} 
\ \ 
+ O\bigl(\cutoff^{-1/p}\bigr)
\end{align}
and the coefficients $A_{pjq}$ can be expressed in 
terms of derivatives of $f$ and $R$ at~$x_{pj}$. 
\end{lemma}
\begin{proof} 
(Sketch.)
Lemma~\ref{lemm:asym1} and the techniques of 
Section II.3 in\citte{wongbook}
show that the contribution from a 
sufficiently small neighbourhood of $x_{pj}$
is $I_{pj}$~\eqref{lemma2}. 
The techniques in 
Section II.3 in\citte{wongbook}
further show that the contributions from 
outside these small neighbourhoods are~$o(1)$.\hfill$\blacksquare$
\end{proof}
Note that $K_{1,0}=\pi$. 
This will be used to choose 
$\rho(\cutoff)=2\pi$ in~\eqref{eq:eta-ga-1}.

\bibliographystyle{cj}

\bibliography{scpaperrefs} 

\end{document}